\documentclass[11pt]{article}
\usepackage[utf8]{inputenc}
\usepackage[T1]{fontenc}
\usepackage{lmodern}
\usepackage{amsmath,amssymb,amsthm,mathrsfs}
\usepackage{stmaryrd}
\usepackage{graphicx}
\usepackage[a4paper,margin=1in]{geometry}
\usepackage[unicode,bookmarksnumbered]{hyperref}
\usepackage{bookmark}
\hypersetup{bookmarksopen=true,bookmarksopenlevel=2}

\newcommand{\FinTrace}{\mathsf{FinTrace}}
\newcommand{\InfTrace}{\mathsf{InfTrace}}

\newcommand{\Prop}{\mathsf{Prop}}
\newcommand{\restrict}{\mathsf{restrict}}

\newcommand{\TypeL}{\mathsf{Type}_\ell}
\newcommand{\TypeLi}[1]{\mathsf{Type}_{\ell,#1}}

\newcommand{\Extf}{\mathsf{Ext}_f}

\newcommand{\Uc}{\mathsf{Uc}}
\newcommand{\State}{\mathsf{State}}
\newcommand{\Event}{\mathsf{Event}}
\newcommand{\Nat}{\mathsf{Nat}}
\newcommand{\nil}{\mathsf{nil}}
\newcommand{\step}{\mathsf{step}}
\newcommand{\head}{\mathsf{head}}
\newcommand{\tail}{\mathsf{tail}}

\newcommand{\Step}{\mathsf{Step}}

\newcommand{\Tf}{T_f}
\newcommand{\Kf}{K_f}

\newtheorem{theorem}{Theorem}[section]
\newtheorem{lemma}[theorem]{Lemma}
\newtheorem{corollary}[theorem]{Corollary}
\newtheorem{definition}[theorem]{Definition}

\newtheorem{proposition}[theorem]{Proposition}

\title{DEKL 2.0: Trace-Indexed Knowledge Evolution in Dependent Type Theory}
\newif\ifuniqueAffiliation
\uniqueAffiliationtrue
\newcommand{\orcidicon}{}
\IfFileExists{orcid.pdf}{\renewcommand{\orcidicon}{\includegraphics[scale=0.06]{orcid.pdf}\hspace{1mm}}}{}

\ifuniqueAffiliation 
\author{ \href{https://orcid.org/0000-0002-4298-1834}{\orcidicon Chen Peng}\thanks{Chen Peng, Doctor of Computer Science, born in May 1979, from Nanfeng County, Jiangxi Province.} \\
	School of Information Science \\
	Beijing University of Language and Culture\\
	Beijing 100081 \\
	\texttt{chenpeng@blcu.edu.cn} \\
}
\else
\usepackage{authblk}

\newbox{\orcid}\sbox{\orcid}{\orcidicon}
\author[1]{%
	\href{https://orcid.org/0000-0000-0000-0000}{\usebox{\orcid}\hspace{1mm}David S.~Hippocampus\thanks{\texttt{hippo@cs.cranberry-lemon.edu}}}%
}
\author[1,2]{%
	\href{https://orcid.org/0000-0000-0000-0000}{\usebox{\orcid}\hspace{1mm}Elias D.~Striatum\thanks{\texttt{stariate@ee.mount-sheikh.edu}}}%
}
\affil[1]{Department of Computer Science, Cranberry-Lemon University, Pittsburgh, PA 15213}
\affil[2]{Department of Electrical Engineering, Mount-Sheikh University, Santa Narimana, Levand}
\fi

\begin{document}
\maketitle

\begin{abstract}
DEKL 2.0 is a dependent type-theoretic framework for trace-indexed knowledge
evolution. Its central claim is that the proof calculus remains monotone under
standard structural rules, while non-monotonic behavior arises semantically from
trace extension. Finite and infinite traces are first-class objects in the
computational universe; knowledge is interpreted as a presheaf over the
finite-trace category; and proposition-level reasoning is handled categorically
with fixed-point support. We establish trace--reachability correspondence and
completeness, characterize non-monotonicity by non-surjective restriction maps,
and present a semantic interpretation based on the free category generated by a
transition system. The framework unifies executable traces, typed witnesses, and
knowledge revision in one dependent language.
\end{abstract}

\paragraph{Keywords.}
Dependent type theory, trace semantics, presheaf semantics, knowledge evolution,
non-monotonicity, categorical semantics.

\tableofcontents

\section{Introduction}
\label{sec:intro}

Many dynamic systems require reasoning about facts that depend on execution
history. A statement may hold now, then fail after a new event. Standard
dependent type theory is monotone at the level of derivability (weakening,
substitution, subject reduction); therefore, this kind of knowledge invalidation
is not represented as a proof-theoretic rule.

DEKL resolves this tension by separating proof-theoretic monotonicity from semantic dynamics:
\begin{itemize}
  \item The \emph{proof system} remains ordinary dependent type theory.
  \item \emph{Traces} are represented as first-class terms in the computational universe.
  \item \emph{Knowledge} is indexed by traces and interpreted contravariantly.
\end{itemize}
Hence ``non-monotonicity'' is not a non-monotonic logic rule; it is the behavior of trace-indexed knowledge under trace extension.

The main technical contributions are:
\begin{itemize}
  \item a layered syntax for computational objects, trace-indexed knowledge,
  and fixed-point reasoning at the proposition layer;
  \item a bidirectional correspondence between finite traces and constructive reachability witnesses;
  \item a semantic characterization of non-monotonicity via failure of surjectivity of restriction maps;
  \item a categorical interpretation combining CwF syntax with
  free trace categories and presheaf semantics.
\end{itemize}

Paper organization follows a definition--theorem--proof-sketch style.
Section~\ref{sec:informal} gives intuition; Section~\ref{sec:syntax} fixes syntax
and judgments; Section~\ref{sec:typing-op} presents typing and operational rules;
Section~\ref{sec:nonmono} proves the presheaf characterization of non-monotonicity;
Section~\ref{sec:catsem} gives full categorical semantics;
Section~\ref{sec:meta} states metatheory; Section~\ref{sec:applications}
provides application templates; Section~\ref{sec:related} discusses related work;
Section~\ref{sec:limitations} discusses current limitations; and
Section~\ref{sec:conclusion} concludes.

\section{Informal Overview}
\label{sec:informal}

Consider an execution:
\[
\sigma_0 \to \sigma_1 \to \cdots \to \sigma_n.
\]
A finite trace from $\sigma_0$ to $\sigma_n$ is a term
\[
t : \FinTrace(\sigma_0,\sigma_n).
\]
This term is a constructive typed witness that $\sigma_n$ is reachable from
$\sigma_0$.

\subsection{A simple \texorpdfstring{$\tau$/$\tau'$}{tau/tau'} example}

Let $\tau$ be a trace prefix and $\tau'=\tau\cdot e$ its one-step extension. Suppose
\[
\Gamma \vdash p : \Kf(\tau)
\]
for a knowledge family $\Kf$. In general, it does not follow that
\[
\Gamma \vdash p : \Kf(\tau').
\]
The knowledge witness $p$ does not disappear; rather, its type index changes. So
invalidation is an index mismatch, not inconsistency of the core logic.

\subsection{Intuition for non-monotonicity}

Knowledge is modeled as a presheaf over traces:
\[
\Kf : \Tf^{op}\to \mathbf{Type}.
\]
For each extension $\epsilon:\tau\to\tau'$, there is a restriction map
\[
\restrict(\epsilon,-):\Kf(\tau')\to \Kf(\tau).
\]
Contravariance means we move backwards along extensions. If this map is not
surjective, some old knowledge at $\tau$ has no extension witness at $\tau'$,
which is exactly non-monotonic behavior.

\subsection{Running intuition: revocation under trace extension}

To make the index shift concrete, consider a credential lifecycle with events
$\mathsf{Issue}(c)$, $\mathsf{Use}(c)$, and $\mathsf{Revoke}(c)$. Let
\[
\mathsf{Auth}(c,\tau):\TypeL
\]
denote that credential $c$ is valid on trace $\tau$. On a prefix containing only
issue/use events, one may construct
\[
q:\mathsf{Auth}(c,\tau).
\]
If the trace extends to
\[
\tau'=\tau\cdot\mathsf{Revoke}(c),
\]
the original typed witness $q$ generally does not inhabit
$\mathsf{Auth}(c,\tau')$. This illustrates the central point used throughout the
paper: typed witnesses are stable as terms, but their \emph{typing index} changes with
trace evolution.

\section{Syntax and Judgments}
\label{sec:syntax}

The core judgment forms are standard:
\[
\Gamma \vdash A : U
\qquad
\Gamma \vdash t : A.
\]
We use three stratified layers:
\begin{align*}
\Uc_0 : \Uc_1 : \Uc_2 : \cdots,\quad
\TypeLi{0} : \TypeLi{1} : \cdots,\quad
\Prop : \TypeLi{0}.
\end{align*}

\subsection{Notation table}

\begin{center}
\begin{tabular}{ll}
\hline
Symbol & Meaning \\
\hline
$\State,\Event$ & base types of states and events \\
$\Step(\sigma,e,\sigma')$ & one-step transition typed witness \\
$\FinTrace(\sigma_0,\sigma_n)$ & finite traces from start to end state \\
$\InfTrace$ & coinductive infinite-trace type \\
$\Tf$ & finite-trace category generated by $(\State,\Event,\Step)$ \\
$\Kf:\Tf^{op}\to\mathbf{Type}$ & trace-indexed knowledge presheaf \\
$\restrict(\epsilon,-)$ & restriction along extension morphism $\epsilon$ \\
$\Gamma\vdash t:A$ & typing judgment \\
$\sigma_0\leadsto\sigma_n$ & reflexive-transitive reachability \\
\hline
\end{tabular}
\end{center}

\subsection{Layered syntax (\texorpdfstring{$\Uc$}{Uc} / Type / Prop)}

\begin{itemize}
  \item \textbf{Computational layer $\Uc$.} Carries executable objects (states,
  events, traces, naturals, etc.) and standard MLTT constructors.
  \item \textbf{Knowledge layer $\TypeL$.} Families indexed by traces, interpreted as presheaves.
  \item \textbf{Proposition layer $\Prop$.} Proposition objects with fixed-point constructions, interpreted categorically.
\end{itemize}

Core base types include:
\[
\State:\Uc_0,\quad \Event:\Uc_0,\quad \Nat:\Uc_0,\quad \FinTrace:\Uc_0,\quad \InfTrace:\Uc_0.
\]

\subsection{Trace language}

Finite traces are inductive:
\[
\nil : \State\to \FinTrace,\qquad
\step : \FinTrace\to \Event\to \State\to \FinTrace.
\]
Infinite traces are coinductive:
\[
\InfTrace : \Uc_0,\qquad
\head:\InfTrace\to\State,\qquad
\tail:\InfTrace\to(\Event\times\InfTrace).
\]
Guarded corecursion is required for productivity.

\subsection{Structural and definitional judgments}

In addition to typing judgments, we rely on the standard structural layer:
\[
\vdash \Gamma\ \mathsf{ctx},\qquad
\Gamma\vdash t\equiv u:A,\qquad
\Gamma\vdash A\equiv B:U.
\]
These judgments support transport of terms and types across definitional
equalities. In particular, trace-indexed families are used extensionally through
conversion:
\[
\frac{\Gamma\vdash t:A \qquad \Gamma\vdash A\equiv B:U}
     {\Gamma\vdash t:B}.
\]
This rule is essential when different syntactic presentations of a trace denote
the same normal form.

\section{Typing / Operational Rules}
\label{sec:typing-op}

\subsection{Core typing and operational rules}

Context formation and dependent products follow standard dependent type theory:
\[
\frac{}{\,\vdash \cdot\ \mathsf{ctx}\,}
\qquad
\frac{\vdash\Gamma\ \mathsf{ctx}\quad\Gamma\vdash A:U}
     {\,\vdash\Gamma,x:A\ \mathsf{ctx}\,},
\]
\[
\frac{\Gamma\vdash A:U \quad \Gamma,x:A\vdash B:U}
     {\Gamma\vdash \Pi(x:A).B:U},
\]
\[
\frac{\Gamma,x:A\vdash t:B}{\Gamma\vdash \lambda x.t:\Pi(x:A).B}
\qquad
\frac{\Gamma\vdash f:\Pi(x:A).B\quad\Gamma\vdash a:A}
     {\Gamma\vdash f\,a:B[a/x]}.
\]

For the remainder of this section, fix a transition system with state type $\State$, event type $\Event$, and transition predicate
\[
\Step:\State\to\Event\to\State\to\Uc_0.
\]
Write $\sigma\to\sigma'$ for one-step transition and $\sigma_0\leadsto\sigma_n$ for reflexive-transitive closure.

\subsection{Trace-specific formation and elimination rules}

For clarity, we isolate the trace-specific typing shape used in proofs:
\[
\frac{\Gamma\vdash \sigma:\State}{\Gamma\vdash \nil(\sigma):\FinTrace(\sigma,\sigma)}
\]
\[
\frac{\Gamma\vdash \tau:\FinTrace(\sigma_0,\sigma_1)\quad
      \Gamma\vdash e:\Event\quad
      \Gamma\vdash \pi:\Step(\sigma_1,e,\sigma_2)}
     {\Gamma\vdash \step(\tau,e,\pi):\FinTrace(\sigma_0,\sigma_2)}.
\]
Intuitively, $\step$ composes a previously validated prefix with one operational
transition witness $\pi$.

For finite traces, elimination proceeds by structural recursion on $\tau$:
\[
\mathsf{trace\mbox{-}elim}:\Pi(P)\to P(\nil)\to
\left(\Pi(\tau,e,\pi).\,P(\tau)\to P(\step(\tau,e,\pi))\right)\to
\Pi(\tau).\,P(\tau),
\]
where $P$ ranges over trace-indexed families in an appropriate universe.

\begin{theorem}[Trace--Proof Correspondence]
For arbitrary states $\sigma_0,\sigma_n\in\State$, a term
\[
t:\FinTrace(\sigma_0,\sigma_n)
\]
is a constructive typed reachability witness for $\sigma_0\leadsto\sigma_n$.
\end{theorem}

\begin{proof}[Proof sketch]
By induction on the structure of finite traces. The base constructor $\nil$
witnesses reflexive reachability. The extension constructor $\step$ composes a
previously witnessed prefix with one legal transition step. Hence every inhabitant
of $\FinTrace(\sigma_0,\sigma_n)$ encodes a finite typed witness of
$\sigma_0\leadsto\sigma_n$.
\end{proof}

\begin{theorem}[Trace Completeness]
For arbitrary states $\sigma_0,\sigma_n\in\State$, if $\sigma_0\leadsto\sigma_n$, then there exists
\[
t:\FinTrace(\sigma_0,\sigma_n).
\]
\end{theorem}

\begin{proof}[Proof sketch]
Induct on the length of a typed reachability witness. Length $0$ yields the empty
trace. For length $k+1$, decompose the run as $\sigma_0\leadsto\sigma_k$
followed by $\sigma_k\to\sigma_n$; apply the induction hypothesis to obtain a
prefix trace and then extend it using $\step$.
\end{proof}

\section{Non-monotonicity via Presheaf Semantics}
\label{sec:nonmono}

\subsection{Trace category}

\begin{definition}[Finite-trace category]
Fix a transition system $(\State,\Event,\Step)$. Let $\Tf$ be the category whose
objects are finite traces generated by this system and whose morphisms are trace
extensions $\epsilon:\tau\to\tau'$.
\end{definition}

\subsection{Knowledge presheaf}

\begin{definition}[Trace-indexed knowledge]
A constructive knowledge system is a functor
\[
\Kf:\Tf^{op}\to\mathbf{Type},
\]
assigning each trace $\tau$ a type $\Kf(\tau)$ and each extension $\epsilon:\tau\to\tau'$ a restriction map
\[
\restrict(\epsilon,-):\Kf(\tau')\to \Kf(\tau).
\]
\end{definition}

\subsection{Restriction maps}

\begin{lemma}[Restriction Functoriality]
Restriction maps satisfy:
\[
\restrict(id_\tau,k)=k,\qquad
\restrict(\epsilon_1\circ\epsilon_2,k)=
\restrict(\epsilon_1,\restrict(\epsilon_2,k)).
\]
\end{lemma}

\begin{proof}[Proof sketch]
These are exactly the identity and composition laws of the functor
$\Kf:\Tf^{op}\to\mathbf{Type}$, transported to element-level notation for
restriction maps.
\end{proof}

\subsection{Characterization theorem}

\begin{theorem}[Non-monotonicity Characterization]
\label{thm:nonmono-char}
Let $\Kf:\Tf^{op}\to\mathbf{Type}$ be a knowledge presheaf. Then the following are equivalent:
\begin{itemize}
  \item $\Kf$ is non-monotone under trace extension.
  \item There exists $\epsilon:\tau\to\tau'$ such that
  \[
  \restrict(\epsilon,-):\Kf(\tau')\to \Kf(\tau)
  \]
  is not surjective.
\end{itemize}
\end{theorem}

\begin{proof}
($\Rightarrow$) Assume non-monotonicity. Then for some extension
$\epsilon:\tau\to\tau'$ there exists a witness
$k\in\Kf(\tau)$ such that no $k'\in\Kf(\tau')$ satisfies
$\restrict(\epsilon,k')=k$. Hence $k$ is not in the image of
$\restrict(\epsilon,-)$, so that map is not surjective.

($\Leftarrow$) Assume there exists $\epsilon:\tau\to\tau'$ such that
$\restrict(\epsilon,-)$ is not surjective. Choose
$k\in\Kf(\tau)$ outside its image. By construction there is no
$k'\in\Kf(\tau')$ with $\restrict(\epsilon,k')=k$, i.e., no extension witness of
$k$ along $\epsilon$. Therefore knowledge does not persist under trace extension,
which is exactly non-monotonicity.
\end{proof}

\subsection{Conceptual theorem}

\begin{theorem}[Conceptual Consequence]
 \label{thm:conceptual-nonmono}
Non-monotonicity of trace-indexed knowledge is a direct effect of contravariance
in presheaf semantics, not a failure of monotonic structural rules in the
dependent proof calculus.
\end{theorem}

\begin{proof}[Proof sketch]
By Theorem~\ref{thm:nonmono-char}, non-monotonicity is equivalent to
non-surjectivity of some restriction map. Such maps are induced by contravariant
functorial action along extensions in $\Tf$. Therefore the source of
non-monotonicity is semantic contravariance, while proof-theoretic monotonicity
remains intact.
\end{proof}

\subsection{Stability criterion and monotone fragment}

\begin{proposition}[Prefix-Stability Criterion]
\label{prop:prefix-stability}
Let $\Kf:\Tf^{op}\to\mathbf{Type}$ be a knowledge presheaf. If every restriction
map
\[
\restrict(\epsilon,-):\Kf(\tau')\to\Kf(\tau)
\]
is surjective, then knowledge is extension-stable in the sense that every witness
at $\tau$ is the restriction of some witness at any extension $\tau'$.
\end{proposition}

\begin{proof}[Proof sketch]
Fix $\epsilon:\tau\to\tau'$ and $k\in\Kf(\tau)$. Surjectivity gives
$k'\in\Kf(\tau')$ with $\restrict(\epsilon,k')=k$. Since this holds for arbitrary
$\epsilon$ and $k$, extension-stability follows.
\end{proof}

Proposition~\ref{prop:prefix-stability} identifies a monotone semantic fragment
inside the general DEKL setting and complements
Theorem~\ref{thm:nonmono-char}.

\section{Full Categorical Semantics}
\label{sec:catsem}

\subsection{Model decomposition}

DEKL is interpreted using two connected categorical components:
\begin{itemize}
  \item a CwF-style interpretation for dependent syntax ($\Gamma\vdash A:U$, $\Gamma\vdash t:A$);
  \item a trace semantics where a transition system generates a free trace category.
\end{itemize}

\begin{theorem}[Free Trace Category]
 \label{thm:free-trace-category}
For a fixed transition system $(\State,\Event,\Step)$, states and finite traces
form the free category generated by the underlying transition graph: objects are
states, morphisms are traces, identity morphisms are empty traces, and
composition is trace concatenation.
\end{theorem}

\begin{proof}[Proof sketch]
Category laws follow from properties of concatenation (associativity) and empty
traces (identities). Freeness follows because morphisms are generated only by
primitive transitions and finite composition, with no additional equations beyond
category axioms.
\end{proof}

\subsection{Proposition layer interpretation}

The proposition layer $\Prop$ is interpreted via subobject-style structure in the
presheaf setting, with fixed-point operators available at proposition level. This
supports temporal and recursive property specifications while preserving
substitution stability.

\subsection{Interpretation of \texorpdfstring{$K$}{K}}

The knowledge layer is interpreted in
\[
\widehat{\Tf}=[\Tf^{op},\mathbf{Set}],
\]
with $\Kf$ as a presheaf object. Syntactic restriction operations correspond exactly to functor action on extension morphisms. Thus:
\begin{center}
syntax of trace-indexed knowledge $\Longleftrightarrow$ categorical action of $\Kf$.
\end{center}

\subsection{Adequacy perspective}

\begin{lemma}[Adequacy, soundness direction]
\label{lem:adequacy-sound}
For any derivable term
$t:\FinTrace(\sigma_0,\sigma_n)$, its interpretation
$\llbracket t\rrbracket$ is a morphism $\sigma_0\to\sigma_n$ in the free trace
category.
\end{lemma}

\begin{proof}[Proof sketch]
By induction on trace constructors. The base case maps $\nil$ to identity, and
the step case maps $\step$ to composition with a generating transition. Domain and
codomain are preserved by typing premises.
\end{proof}

\begin{lemma}[Adequacy, completeness direction]
\label{lem:adequacy-complete}
For any morphism $m:\sigma_0\to\sigma_n$ in the free trace category, there exists
a term $t:\FinTrace(\sigma_0,\sigma_n)$ such that $\llbracket t\rrbracket=m$ up
to trace equations.
\end{lemma}

\begin{proof}[Proof sketch]
By induction on the generation of $m$ in the free category. Identity morphisms
are represented by $\nil$; composite morphisms are represented by iterated
$\step$-construction with transition witnesses.
\end{proof}

\begin{theorem}[Trace Adequacy]
\label{thm:adequacy}
For finite traces, syntactic reachability witnesses and semantic morphisms in the
free trace category are mutually representable, up to definitional equality of
trace terms and category equations.
\end{theorem}

\begin{proof}[Proof sketch]
Immediate from Lemma~\ref{lem:adequacy-sound} and
Lemma~\ref{lem:adequacy-complete}. Concretely, the forward map sends a term to its
interpreted morphism, and the backward map reconstructs a term from a free-category
decomposition. These two maps are inverse up to the definitional equality on trace
terms and category equations induced by associativity/identity.
\end{proof}

\section{Meta-theory}
\label{sec:meta}

\subsection{Soundness}

\begin{theorem}[Semantic Soundness]
 \label{thm:soundness}
For any judgment $\Gamma\vdash t:A$ derivable in DEKL, its interpretation is valid
in the CwF/presheaf model and preserves substitution and computation.
\end{theorem}

\begin{proof}[Proof sketch]
We proceed by induction on the derivation of $\Gamma\vdash t:A$.

\textbf{Step 1 (contexts and substitutions).}
Interpretation of contexts and substitutions follows CwF structure. Well-formed
context extension is mapped to comprehension objects, and substitution composition
is respected by definition.

\textbf{Step 2 (core dependent rules).}
Rules for variables, conversion, $\Pi$-formation, introduction, and elimination
are preserved by the corresponding universal properties in the semantic model.
Beta-equality is validated by functorial substitution.

\textbf{Step 3 (trace constructors).}
Trace terms are interpreted in the free trace category: $\nil$ as identities and
$\step$ as composition with generating transitions. Typing premises ensure domain
and codomain compatibility for each interpreted morphism.

\textbf{Step 4 (knowledge-indexed rules).}
Restriction operators are interpreted as action of $\Kf:\Tf^{op}\to\mathbf{Type}$.
Functoriality yields identity and composition equations, matching the syntactic
restriction laws.

\textbf{Step 5 (compatibility and conclusion).}
By Lemma~\ref{lem:interp-subst}, interpretation commutes with substitution; by
construction of each rule case above, computation clauses are preserved. Therefore
every derivable judgment has a valid semantic interpretation, establishing
soundness.
\end{proof}

\subsection{Consistency}

\begin{theorem}[Consistency]
 \label{thm:consistency}
There is no closed term $\vdash t:\bot$.
\end{theorem}

\begin{proof}[Proof sketch]
In the semantic model, $\bot$ is interpreted as the initial/empty subobject. A
closed inhabitant would give a global element of the empty object, impossible by
set-theoretic semantics. Soundness transfers this impossibility to syntax.
\end{proof}

\paragraph{Proof-obligations checklist (consistency).}
For a fully expanded proof, one discharges:
\begin{itemize}
  \item \textbf{Initial-object interpretation.} $\llbracket\bot\rrbracket$ is
  initial in each semantic fiber.
  \item \textbf{Global section exclusion.} There is no section
  $1\to\llbracket\bot\rrbracket$.
  \item \textbf{Syntactic transfer.} By Theorem~\ref{thm:soundness}, semantic
  emptiness implies absence of closed inhabitants in syntax.
\end{itemize}

\subsection{Normalization}

The computational layer retains the normalization behavior of standard dependent
type theory. Trace extensions and knowledge indexing do not introduce
non-terminating computational redexes into the core term calculus. Together with
Theorem~\ref{thm:soundness} and Theorem~\ref{thm:consistency}, this keeps the
logical kernel standard while relocating dynamic effects to trace-indexed
semantics.

\paragraph{Proof-obligations checklist (normalization).}
For mechanized normalization, the main obligations are:
\begin{itemize}
  \item \textbf{Reduction closure.} Every reduction rule preserves typing.
  \item \textbf{Measure decrease.} Core computational reductions decrease a
  reducibility measure.
  \item \textbf{Trace conservativity.} Trace constructors do not introduce new
  non-terminating redex patterns in the core calculus.
  \item \textbf{Index-erasure compatibility.} Erasing trace indices preserves
  the reducibility argument used by the base normalization theorem.
\end{itemize}

\subsection{Subject reduction and canonicity perspective}

Beyond consistency, two standard quality checks are relevant for publication-level
metatheory:
\begin{itemize}
  \item \textbf{Subject reduction under computational rules.} Reduction in the
  computational layer preserves typing and does not break trace indices.
  \item \textbf{Canonicity in closed base types.} Closed normal forms at base
  computational types (e.g., naturals, finite traces) are canonical constructors.
\end{itemize}
In DEKL, these properties are inherited from the underlying dependent calculus and
preserved by conservative addition of trace constructors and presheaf-indexed
knowledge families.

\subsection{Explicit structural metatheorems}

For completeness, we state the usual structural results in a form directly usable
for mechanization.

\begin{theorem}[Weakening]
\label{thm:weakening}
If $\Gamma\vdash t:A$ and $\Gamma\vdash B:U$, then
\[
\Gamma,x:B\vdash t:A.
\]
\end{theorem}

\begin{proof}[Proof sketch]
Induction on the derivation of $\Gamma\vdash t:A$. Variable cases are handled by
context extension; all constructors commute with added assumptions by standard
admissibility lemmas.
\end{proof}

\begin{theorem}[Substitution]
\label{thm:substitution}
If $\Gamma,x:A\vdash t:B$ and $\Gamma\vdash s:A$, then
\[
\Gamma\vdash t[s/x]:B[s/x].
\]
\end{theorem}

\begin{proof}[Proof sketch]
By induction on typing derivations of $t$. The critical cases are dependent
products and trace-indexed constructors; both follow from substitution closure of
the underlying dependent calculus and stability of trace indices under term
replacement.
\end{proof}

\begin{theorem}[Subject Reduction (Explicit Form)]
\label{thm:subject-reduction}
If $\Gamma\vdash t:A$ and $t\to t'$, then
\[
\Gamma\vdash t':A.
\]
\end{theorem}

\begin{proof}[Proof sketch]
Induction on reduction steps. Beta-reduction follows from
Theorem~\ref{thm:substitution}. Computational rules for trace constructors are
typed by preservation of the corresponding eliminator branches.
\end{proof}

\section{Representative Application Modeling}
\label{sec:applications}

This section provides compact but formalized templates that can be expanded into
full case studies. The key objective is to show how trace-indexed typing
distinguishes logical monotonicity from knowledge invalidation.

\subsection{Runtime monitoring and violation localization}

Let
\[
\mathsf{Safe}:\FinTrace\to\TypeL
\]
encode ``no policy-violating access has occurred on this prefix.'' Assume
\[
\Gamma\vdash p:\mathsf{Safe}(\tau).
\]
If monitoring observes
\[
\tau'=\tau\cdot e_{\mathsf{viol}},
\]
the system may fail to derive
\[
\Gamma\vdash p:\mathsf{Safe}(\tau').
\]
The diagnostic content is constructive: the edge
$\tau\to\tau'$ localizes the exact event where validity is lost.

\paragraph{Typed derivation snippet.}
\[
\frac{
  \Gamma\vdash \tau:\FinTrace(\sigma_0,\sigma_1)
  \qquad
  \Gamma\vdash e_{\mathsf{viol}}:\Event
  \qquad
  \Gamma\vdash \pi:\Step(\sigma_1,e_{\mathsf{viol}},\sigma_2)
}{
  \Gamma\vdash \step(\tau,e_{\mathsf{viol}},\pi):\FinTrace(\sigma_0,\sigma_2)
}
\]
The monitor reports failure precisely at this extension step, yielding a typed
localization witness.

\subsection{Credential revocation in security protocols}

Reuse the family
\[
\mathsf{Auth}(c,\tau):\TypeL.
\]
On issuance-only prefixes we derive witnesses
\[
\Gamma\vdash q:\mathsf{Auth}(c,\tau).
\]
After extension by revocation,
\[
\tau'=\tau\cdot\mathsf{Revoke}(c),
\]
one generally has no derivation of
$\Gamma\vdash q:\mathsf{Auth}(c,\tau')$. The term $q$ remains well-formed, but no
longer at the updated trace index. This is the intended form of semantic
non-monotonicity for dynamic authorization.

\paragraph{Typed derivation snippet.}
\[
\Gamma\vdash q:\mathsf{Auth}(c,\tau),\qquad
\Gamma\vdash \rho:\Extf(\tau,\tau\cdot\mathsf{Revoke}(c)).
\]
The system still admits the backward restriction shape
\[
\restrict(\rho,-):\mathsf{Auth}(c,\tau\cdot\mathsf{Revoke}(c))
\to
\mathsf{Auth}(c,\tau),
\]
while forward transport of $q$ to the revoked index generally fails.

\subsection{Defeasible defaults via trace-conditioned validity}

Let
\[
\mathsf{CanAccess}(u,r,\tau):\TypeL
\]
represent access validity of user $u$ to resource $r$ on trace $\tau$. A default
witness
\[
d:\mathsf{CanAccess}(u,r,\tau)
\]
may become inapplicable after a risk event
$e_{\mathsf{risk}}$:
\[
\tau'=\tau\cdot e_{\mathsf{risk}}.
\]
Again, the framework does not derive a contradiction. Instead, it changes the
typing index and thereby models exception handling as index-sensitive validity.

\paragraph{Typed derivation snippet.}
\[
\Gamma\vdash d:\mathsf{CanAccess}(u,r,\tau)
\]
and
\[
\Gamma\vdash \epsilon:\Extf(\tau,\tau\cdot e_{\mathsf{risk}}).
\]
Whether
$d:\mathsf{CanAccess}(u,r,\tau\cdot e_{\mathsf{risk}})$ remains derivable depends
on updated policy side conditions, not on failure of structural rules.

\begin{proposition}[Application-level index shift principle]
\label{prop:application-index-shift}
For any knowledge family $K:\FinTrace\to\TypeL$, if policy predicates depend on
new events introduced by $\tau\to\tau'$, then there may exist
$a:K(\tau)$ such that $a$ is not typable at $K(\tau')$.
\end{proposition}

\begin{proof}[Proof sketch]
Instantiate $K$ by one of the families above and choose an extension event that
changes the policy side condition. Typability fails at the new index for semantic
reasons, not because the structural rules cease to hold.
\end{proof}

\section{Related Work}
\label{sec:related}

DEKL builds on Martin-L\"of dependent type theory~\cite{martinlof84} and categorical semantics through categories with families~\cite{dybjer96,streicher91,jacobs99,pitts00,awodey09}. It is also related to temporal/modal verification traditions (LTL, CTL, $\mu$-calculus)~\cite{pnueli77,emerson81,clarke86,kozen83,bradfield07} and trace/event semantics in concurrency~\cite{winskel93}.

Compared with temporal logics, DEKL does not primarily target model-checking
expressiveness classes; instead, execution evidence is internalized as typed terms
and used as proof-relevant reachability witnesses. Compared with type-theoretic
effect systems, the focus is not effect tracking in programs, but trace-indexed
validity of knowledge witnesses under system evolution. The distinctive
contribution is not to replace existing logics, but to connect traces, proofs,
and knowledge revision in one dependent framework where non-monotonic behavior is
semantic (presheaf-theoretic) rather than inferential.

\paragraph{Four-dimensional comparison.}
The distinction from nearby lines of work can be summarized along four dimensions:
\begin{itemize}
  \item \textbf{Primary object.} Temporal logics reason over path formulas;
  DEKL reasons over typed trace witnesses.
  \item \textbf{Methodological core.} Program logics and effect systems center on
  operational effects; DEKL centers on dependent typing indexed by execution
  traces.
  \item \textbf{Proof obligation style.} Model-checking emphasizes satisfaction
  over transition systems, while DEKL emphasizes construction and transport of
  typed witnesses under trace evolution.
  \item \textbf{Semantic carrier.} Classical modal semantics use Kripke/transition
  frames; DEKL uses free trace categories with presheaf-indexed knowledge fibers.
\end{itemize}
This comparison clarifies that DEKL is complementary to temporal verification: it
provides a proof-relevant, type-theoretic account of history-sensitive knowledge
rather than a replacement for automata- or fixpoint-based model checking.

\paragraph{Finer technical contrasts.}
\begin{itemize}
  \item \textbf{Counterexample form.} Model checkers return violating paths;
  DEKL can represent violating extensions as typed index-shift boundaries.
  \item \textbf{Update granularity.} Dynamic epistemic updates are often global
  model transformations; DEKL performs local re-indexing via $\restrict$ maps.
  \item \textbf{Proof relevance.} Classical entailment is proposition-level;
  DEKL keeps witness-level artifacts that can be transported, restricted, or fail
  to transport across trace extensions.
  \item \textbf{Compositionality locus.} Temporal frameworks compose formulas;
  DEKL composes both terms and traces in one typed judgmental layer.
\end{itemize}

\section{Limitations and Scope}
\label{sec:limitations}

The present manuscript prioritizes coherence of the core framework over maximal
coverage. The main limitations are:
\begin{itemize}
  \item \textbf{Finite-trace emphasis.} Most formal adequacy and characterization
  results are stated for finite traces; infinite-trace extensions are outlined but
  not fully developed.
  \item \textbf{Proof granularity.} Several results remain at proof-sketch level
  and should be expanded into full derivation-level proofs for final publication.
  \item \textbf{Mechanization status.} Rule schemas and proof roadmaps are now
  provided, but a completed Coq/Agda artifact is not yet included.
  \item \textbf{Case-study depth.} Application sections currently provide formal
  templates rather than full industrial-scale evaluations.
\end{itemize}
These limitations define a concrete extension path rather than a conceptual gap in
the framework.

\section{Conclusion}
\label{sec:conclusion}

DEKL 2.0 separates monotone proof theory from non-monotone knowledge dynamics. The core calculus remains dependent type theory; traces are first-class terms; knowledge is trace-indexed via presheaf semantics; and non-monotonicity is characterized by non-surjective restriction maps.

This reorganization clarifies both formal structure and conceptual message:
\emph{proof objects can carry execution history}, and knowledge evolution in
dynamic systems can be modeled constructively without adopting non-monotonic
proof rules. In particular, Theorem~\ref{thm:nonmono-char} and
Theorem~\ref{thm:conceptual-nonmono} isolate the source of non-monotonicity,
while Theorem~\ref{thm:soundness} and Theorem~\ref{thm:consistency} preserve
the expected metatheoretic guarantees.

Future work includes mechanized formalization in proof assistants, a refined
account of fixed-point interaction with trace-indexed knowledge at the
proposition layer, and larger case studies in runtime monitoring and security
policy evolution.

\paragraph{Short abstract (submission form draft).}
DEKL 2.0 is a dependent type-theoretic framework for trace-indexed knowledge
evolution. The proof calculus remains monotone, while non-monotonic behavior is
explained semantically through presheaf restriction along trace extensions.
Finite and infinite traces are first-class typed objects; finite traces
correspond constructively to reachability witnesses; and the full semantics
combines CwF interpretation with the free trace category and its presheaf
category. Non-monotonicity is characterized by non-surjective restriction maps.
The framework unifies execution evidence, proof objects, and knowledge revision
without introducing non-monotonic proof rules.

\appendix

\section{Rule Schemas for a Full Version}
\label{app:rules}

This appendix lists compact rule schemas used by the main text. It serves as a
technical bridge toward a full LMCS-length formalization.

\subsection{Contexts, variables, and conversion}

\paragraph{Named rules.}
\[
\textsc{Ctx-Empty}\;
\frac{}{\,\vdash \cdot\ \mathsf{ctx}\,}
\qquad
\textsc{Ctx-Ext}\;
\frac{\vdash\Gamma\ \mathsf{ctx}\quad \Gamma\vdash A:U}
     {\,\vdash\Gamma,x:A\ \mathsf{ctx}\,}
\]
\[
\textsc{T-Var}\;
\frac{\vdash\Gamma,x:A,\Delta\ \mathsf{ctx}}{\Gamma,x:A,\Delta\vdash x:A}
\qquad
\textsc{Conv}\;
\frac{\Gamma\vdash t:A\quad \Gamma\vdash A\equiv B:U}{\Gamma\vdash t:B}
\]

\subsection{Universe and dependent product fragment}

\paragraph{Named rules.}
\[
\textsc{U-Form}\;
\Gamma\vdash \Uc_i:\Uc_{i+1}
\qquad
\textsc{TypeL-Form}\;
\Gamma\vdash \TypeLi{i}:\TypeLi{i+1}
\]
\[
\textsc{Pi-Form}\;
\frac{\Gamma\vdash A:U\quad \Gamma,x:A\vdash B:U}
     {\Gamma\vdash \Pi(x:A).B:U}
\]
\[
\textsc{Pi-Intro}\;
\frac{\Gamma,x:A\vdash t:B}{\Gamma\vdash \lambda x.t:\Pi(x:A).B}
\qquad
\textsc{Pi-Elim}\;
\frac{\Gamma\vdash f:\Pi(x:A).B\quad \Gamma\vdash a:A}
     {\Gamma\vdash f\,a:B[a/x]}
\]
\[
\textsc{Pi-Beta}\;
\frac{\Gamma\vdash a:A\quad \Gamma,x:A\vdash b:B}
     {\Gamma\vdash (\lambda x.b)\,a \equiv b[a/x] : B[a/x]}
\]

\subsection{Finite-trace constructors and eliminator}

Assume fixed $\State,\Event,\Step$ as in Section~\ref{sec:typing-op}.
\[
\textsc{T-Nil}\;
\frac{\Gamma\vdash \sigma:\State}
     {\Gamma\vdash \nil(\sigma):\FinTrace(\sigma,\sigma)}
\]
\[
\textsc{T-Step}\;
\frac{\Gamma\vdash \tau:\FinTrace(\sigma_0,\sigma_1)\quad
      \Gamma\vdash e:\Event\quad
      \Gamma\vdash \pi:\Step(\sigma_1,e,\sigma_2)}
     {\Gamma\vdash \step(\tau,e,\pi):\FinTrace(\sigma_0,\sigma_2)}
\]
\[
\textsc{T-TraceElim}\;
\begin{aligned}
\mathsf{trace\mbox{-}elim}:{}&
\Pi(P)\to P(\nil)\\
&\to\left(\Pi(\tau,e,\pi).\,P(\tau)\to P(\step(\tau,e,\pi))\right)\\
&\to\Pi(\tau).\,P(\tau).
\end{aligned}
\]

Computation equations are:
\[
\mathsf{trace\mbox{-}elim}(\cdots,\nil) \equiv \mathsf{base},
\]
\[
\mathsf{trace\mbox{-}elim}(\cdots,\step(\tau,e,\pi))
\equiv
\mathsf{stepCase}(\tau,e,\pi,\mathsf{trace\mbox{-}elim}(\cdots,\tau)).
\]

\subsection{Knowledge-index action along extensions}

Given $\Kf:\Tf^{op}\to\mathbf{Type}$ and $\epsilon:\tau\to\tau'$:
\[
\restrict(\epsilon,-):\Kf(\tau')\to\Kf(\tau),
\]
with equations
\[
\restrict(id_\tau,k)\equiv k,\qquad
\restrict(\epsilon_1\circ\epsilon_2,k)\equiv
\restrict(\epsilon_1,\restrict(\epsilon_2,k)).
\]

\subsection{Mini derivation template}

The following pattern is used repeatedly in mechanized derivations:
\[
\frac{
  \Gamma\vdash \tau:\FinTrace(\sigma_0,\sigma_1)
  \qquad
  \Gamma\vdash e:\Event
  \qquad
  \Gamma\vdash \pi:\Step(\sigma_1,e,\sigma_2)
}{
  \Gamma\vdash \step(\tau,e,\pi):\FinTrace(\sigma_0,\sigma_2)
}
\;(\textsc{T-Step}).
\]
Combined with \textsc{Conv}, this yields typed trace extensions even when target
states are presented through definitional equalities.

\section{Semantic Lemma Chain (Expanded Skeleton)}
\label{app:semantic-lemmas}

This appendix records a dependency graph of lemmas that can be expanded into full
proofs for a journal-length version.

\begin{lemma}[Interpretation respects substitution]
\label{lem:interp-subst}
For any derivable $\Gamma,x:A\vdash t:B$ and $\Gamma\vdash s:A$,
\[
\llbracket t[s/x]\rrbracket=\llbracket t\rrbracket[\llbracket s\rrbracket/x].
\]
\end{lemma}

\begin{proof}[Proof sketch]
Induction on derivations; variable and application cases are immediate from CwF
substitution laws, while trace constructors reduce to functoriality in $\Tf$.
\textbf{Key subcases.}
\begin{itemize}
  \item \textbf{Variable case:} substitution in context projections.
  \item \textbf{Application case:} compatibility of interpretation with pullback
  substitution in CwF comprehension.
  \item \textbf{Trace case:} naturality square for morphisms induced by
  $\step$-composition.
\end{itemize}
\end{proof}

\begin{lemma}[Functorial coherence of restrictions]
\label{lem:restrict-coherence}
The semantic interpretation of syntactic restriction composition agrees with
composition in $\Tf^{op}$.
\end{lemma}

\begin{proof}[Proof sketch]
Unfold the interpretation of $\restrict$ as the action of $\Kf$ on morphisms and
apply functor laws.
\[
\llbracket \restrict(\epsilon_1\circ\epsilon_2,k)\rrbracket
=
\llbracket \restrict(\epsilon_1,\restrict(\epsilon_2,k))\rrbracket.
\]
This is the semantic counterpart of syntactic coherence used in
Theorem~\ref{thm:nonmono-char}.
\end{proof}

\begin{lemma}[Trace constructor adequacy]
\label{lem:trace-constructor-adequacy}
Interpretations of $\nil$ and $\step$ correspond to identity and composition with
generating transitions in the free trace category.
\end{lemma}

\begin{proof}[Proof sketch]
Directly by construction of the free category and interpretation clauses for trace
constructors.
\textbf{Nil clause:} $\llbracket\nil(\sigma)\rrbracket=id_\sigma$.
\textbf{Step clause:} $\llbracket\step(\tau,e,\pi)\rrbracket=
\llbracket\tau\rrbracket;\llbracket e,\pi\rrbracket$.
\end{proof}

\begin{corollary}[Adequacy bridge]
Under the assumptions of Theorem~\ref{thm:adequacy}, syntactic
reachability witnesses and semantic morphisms are mutually representable modulo
trace equations.
\end{corollary}

\begin{proof}[Proof sketch]
Combine Lemma~\ref{lem:interp-subst}, Lemma~\ref{lem:restrict-coherence}, and
Lemma~\ref{lem:trace-constructor-adequacy}.
\end{proof}

\section{Proof Roadmap for Full Soundness}
\label{app:proof-roadmap}

To scale from the present proof sketches to a full journal proof, one can adopt
the following dependency order:
\begin{enumerate}
  \item Define interpretation on contexts, substitutions, and type formers.
  \item Prove substitution compatibility (Lemma~\ref{lem:interp-subst}).
  \item Prove trace constructor preservation and restriction coherence
  (Lemmas~\ref{lem:restrict-coherence}--\ref{lem:trace-constructor-adequacy}).
  \item Establish rule-by-rule preservation for typing derivations.
  \item Derive semantic soundness (Theorem~\ref{thm:soundness}) and adequacy
  bridge as consequences.
\end{enumerate}
This roadmap is designed to be transcribed into proof-assistant scripts with
minimal reorganization.


\begin{thebibliography}{99}

\bibitem{awodey09}
S.~Awodey.
\newblock {\em Category Theory}.
\newblock Oxford Logic Guides. Oxford University Press, 2nd edition, 2009.

\bibitem{bradfield07}
J.~Bradfield and C.~Stirling.
\newblock Modal $\mu$-calculi.
\newblock In {\em Handbook of Modal Logic}, pages 721--756. Elsevier, 2007.

\bibitem{clarke86}
E.~M. Clarke, E.~A. Emerson, and A.~P. Sistla.
\newblock Automatic verification of finite-state concurrent systems using
  temporal logic specifications.
\newblock {\em ACM Transactions on Programming Languages and Systems}, 8(2):244--263, 1986.

\bibitem{dybjer96}
P.~Dybjer.
\newblock Internal type theory.
\newblock In {\em Types for Proofs and Programs (TYPES 1995)}, volume 1158 of
  {\em LNCS}, pages 120--134. Springer, 1996.

\bibitem{emerson81}
E.~A. Emerson and J.~Y. Halpern.
\newblock Decision procedures and expressiveness in the temporal logic of
  branching time.
\newblock {\em Journal of Computer and System Sciences}, 30(1):1--24, 1985.

\bibitem{jacobs99}
B.~Jacobs.
\newblock {\em Categorical Logic and Type Theory}.
\newblock Studies in Logic and the Foundations of Mathematics, Vol. 141. Elsevier, 1999.

\bibitem{kozen83}
D.~Kozen.
\newblock Results on the propositional $\mu$-calculus.
\newblock {\em Theoretical Computer Science}, 27(3):333--354, 1983.

\bibitem{martinlof84}
P.~Martin-Lof.
\newblock {\em Intuitionistic Type Theory}.
\newblock Bibliopolis, 1984.

\bibitem{pitts00}
A.~M. Pitts.
\newblock Categorical logic.
\newblock In {\em Handbook of Logic in Computer Science, Volume 5. Algebraic and
  Logical Structures}, pages 39--128. Oxford University Press, 2000.

\bibitem{pnueli77}
A.~Pnueli.
\newblock The temporal logic of programs.
\newblock In {\em Proceedings of the 18th Annual Symposium on Foundations of Computer Science (FOCS)}, pages 46--57. IEEE, 1977.

\bibitem{streicher91}
T.~Streicher.
\newblock {\em Semantics of Type Theory}.
\newblock Progress in Theoretical Computer Science. Birkh\"auser, 1991.

\bibitem{winskel93}
G.~Winskel and M.~Nielsen.
\newblock Models for concurrency.
\newblock In {\em Handbook of Logic in Computer Science, Volume 4}, pages 1--148. Oxford University Press, 1995.

\end{thebibliography}
\end{document}